\documentclass{article}
\usepackage{amsmath}
\usepackage{vmargin}
\usepackage{amssymb}
\usepackage{stmaryrd}
\usepackage{graphicx}
\usepackage{bussproofs}
\usepackage{tensor}
\usepackage{xypic}
\usepackage{amsthm}
\usepackage{url}

\newcommand{\defeq}{\triangleq}

\newcommand{\ringunit}{\textbf 1}
\newcommand{\ringzero}{\textbf 0}
\newcommand{\aff}{\mathrm{Aff}_1^c}

\newcommand{\lexpt}{\mathsf{exp}}
\newcommand{\lvart}{\mathsf{var}}
\newcommand{\lcomt}{\mathsf{com}}
\newcommand{\lacct}{\mathsf{acc}}

\newcommand{\lcompc}{\mathsf{comp}}
\newcommand{\lopc}{\mathsf{op}}

\newcommand{\lifc}{\mathsf{if}}
\newcommand{\lnewc}{\mathsf{new}}
\newcommand{\lskipc}{\mathsf{skip}}

\newcommand{\sbr}[1]{\llbracket {#1}\rrbracket}
\newtheorem{theorem}{Theorem}
\newtheorem{definition}[theorem]{Definition}
\newtheorem{lemma}[theorem]{Lemma}
\newtheorem{proposition}[theorem]{Proposition}
\newtheorem{example}[theorem]{Example}

\newcommand{\mv}[1]{\mathsf{#1}}
\newcommand{\keep}{\downharpoonright}

\newcommand{\catplus}{\varoplus}
\newcommand{\cattimes}{\varodot}

\newcommand{\constr}[2]{#1 \blacktriangleright #2}
\newcommand{\true}{\mathit{true}}

\begin{document}
\title{From bounded affine types to automatic timing analysis}
\author{
Dan R. Ghica \qquad Alex Smith\\
University of Birmingham
}
\maketitle
\begin{abstract}
Bounded linear types have proved to be useful for automated resource analysis and control in functional programming languages. In this paper we introduce an affine bounded linear typing discipline on a general notion of resource which can be modeled in a semiring. For this type system we provide both a general type-inference procedure, parameterized by the decision procedure of the semiring equational theory, and a (coherent) categorical semantics. This is a very useful type-theoretic and denotational framework for many applications to resource-sensitive compilation, and it represents a generalization of several existing type systems. As a non-trivial instance, motivated by our ongoing work on hardware compilation, we present a complex new application to calculating and controlling timing of execution in a (recursion-free) higher-order functional programming language with local store. 
\end{abstract}

\section{Resource-aware types and semantics}

The two important things about a computer program are what it computes and what resources it needs to carry out the computation successfully. Correctness of the input-output behavior of programs has been, of course, the object of much research from various conceptual angles: logical, semantical, type-theoretical and so on. Resource analysis has been conventionally studied for algorithms, such as time and space complexity, and for programs has long been a part of research in compiler optimization. 

An exciting development was the introduction of semantic~\cite{DBLP:conf/concur/Boudol93} and especially type theoretic~\cite{DBLP:conf/lics/Hofmann99a} characterizations of resource consumption in functional programming languages. Unlike algorithmic analyses, type based analysis are formal and can be statically checked for implementations of algorithms in concrete programming languages. Unlike static analysis, a typing mechanism is compositional which means that it supports, at least in principle, separate compilation and even a foreign function interface: it is an analysis based on signatures rather than implementations. 

Linear logic and typing, because of the fine-grained treatment of resource-sensitive structural rules, constitute an excellent framework for resource analysis, especially in its bounded fragment~\cite{girard1992bounded}, which can logically characterize polynomial time computation. Bounded Linear Logic (BLL) was subsequently extended to improve its flexibility while retaining poly-time~\cite{DBLP:conf/tlca/LagoH09} and further extensions to linear \emph{dependent} typing were used to completely characterize complexity of evaluation of functional programs \cite{DBLP:journals/corr/abs-1104-0193}. 

Although such analyses use \emph{time} as a motivating example, they can be readily adapted to other \emph{consumable} resources such as energy or network traffic. A slightly different angle on resource sensitivity is control of \emph{reusable} resources which can be allocated and de-allocated at runtime, the typical example of which is \emph{memory}, especially \emph{local} (stack-allocated) memory. A well-behaved program will leave the stack empty upon termination, so talking about the total usage of stack-allocated memory is meaningless. Also, talking about the total number of allocations (\emph{push}) on the stack is rarely interesting. What is interesting is that the maximum size of the stack, which is bounded on most architectures, is not exceeded. For reusable resources the relevant limits are, therefore, concerning the \emph{rate} at which the resource is used, for example power (as opposed to energy) or bandwidth (as opposed to total network traffic). In previous work, the first author used a BLL-like type system to bound the number of simultaneous concurrent threads in a parallel functional programming language in order to extract finite models~\cite{DBLP:journals/tcs/GhicaMO06}. This view of concurrent threads as a (reusable) resource proved to be instrumental in facilitating the compilation of functional-imperative programming languages directly into electronic circuits~\cite{DBLP:conf/popl/Ghica07} and is closely related (conceptually, if not formally) to the use of sub-linear runtime space restrictions~\cite{DBLP:conf/aplas/LagoS10}. 

As type systems become more sophisticated the burden on the programmer may increase correspondingly, unless type inference is used to automate the typing process. In the case of bounded linear types the bounds can be calculated fully automatically, by solving a system of numeric constraints~\cite{DBLP:conf/popl/GhicaS11}. In the case of dependent typing this procedure is not decidable, but reduction to constraint systems can still greatly simplify the typing burden~\cite{DBLP:conf/popl/LagoP13}. 

Resource-awareness can be usually captured quite well by operational models of programming languages or typing systems. This is a common feature of the work cited above. A notable exception is the use of game semantics as a \emph{denotational} framework for resource sensitivity, which was introduced by the first author~\cite{DBLP:conf/popl/Ghica05} and recently formulated in a more abstract denotational setting~\cite{lairdmmp13}. 

\section{Contribution and paper outline}

The first part of our present work generalizes bounded linear (or, rather, affine) type systems to an abstract notion of resource, so long as it can be modeled in a semiring. For this abstract type system we show how the problem of type inference can be reduced to a system of constraints based on the equational theory of the resource semiring. Provided this theory is decidable, a type inference algorithm automatically follows. Also for the abstract type system we give a simple categorical framework for which we prove the key result of \emph{coherence}. Because meaning is calculated inductively on the derivation of the typing judgment, and because these derivations are not unique, coherence is the property guaranteeing that all these interpretations are actually equal. Coherence for a categorical semantics is the analogue of a subject reduction lemma in an operational semantics, the basic guarantee of its well-formedness. 

The second part of our work presents a non-trivial application to timing analysis and automated pipelining of computations in a recursion-free functional programming language with local store. The key notion of resource is that of a \emph{schedule} of a computation, i.e.\ the multiset of \emph{stages}, as defined by the start and end of computation, at which a term undergoes execution. Mathematically, stages are contractive affine transformations representing a sub-interval of the unit interval, taken conventionally as the overall duration of execution of the entire program. The resource reading of duration makes good intuitive sense in our target application, automated pipelining, as each stage in a pipeline can be seen as a reusable resource which is either free or busy at any given time. Both the type inference and the categorical semantics are applicable to a variety of resource-sensitive type systems and semantics, generalizing prior work such as~\cite{DBLP:conf/popl/GhicaS11}. 

Finally, we give a game-semantic model for the (concrete) type system in order to justify it computationally. The game-semantic model is denotational therefore compositional by construction, and the categorical semantics ensures that it provides a reasonable interpretation. We do not provide a conventional operational semantics because the game semantics provides enough operational content to be directly usable in the definition of a compiler as proved practically by our previous work on hardware synthesis (\emph{loc.\ cit.}) and more formally in forthcoming work on constructing abstract machines from game semantics~\cite{fredrikssong13}. Moreover, the game semantics   provides an immediate model for foreign function interfaces, which is essential in the development of a useful compiler~\cite{DBLP:conf/memocode/Ghica11}.

\subsection{Related work}

The problem of calculating timing bounds for program execution has been studied extensively. In functional languages it is especially relevant for reactive~\cite{wan2001real} and syn\-chronous \cite{pilaud1987lustre} programming. A variety of methods have been proposed, from static analysis~\cite{liu1998automatic} to full dependent types~\cite{Crary:2000:RBC:325694.325716}. The defining feature of our work is the fact that it is type-based and offers fully automated inference, so requires no annotations or additional effort from the programmer. The application to pipelining is also suitable in terms of our restriction to recursion-free programming, as pipelining is most commonly used as an optimization for finite unfolding of recursive (or iterative) terms. 

\section{Bounded affine types, a general framework}\label{sub:agf}
 Types are generated by the grammar $\theta::=\sigma\mid (J\cdot\theta)\multimap\theta$, where $\sigma$ is a fixed collection of base types and $J\in\mathcal J$, where $(\mathcal J,+,\times,\ringzero,\ringunit)$ is a semiring. We will always take $\cdot$ to bind strongest so we will omit the brackets. 

Let $\Gamma= x_1{:}J_1{\cdot}\theta_1,\ldots,x_n{:}J_n{\cdot} \theta_n$ be a list of identifiers $x_i$ and types $\theta_i$, annotated with semiring elements $J_i$.  Let $fv(M)$ be the set of free variables of term $M$, defined in the usual way. The typing rules are:
\begin{center}
  \AxiomC{ }
  \RightLabel{Identity}
  \UnaryInfC{$x:\mathbf 1\cdot\theta\vdash x:\theta$}
  \DisplayProof\\[1.5ex]
  \AxiomC{$\Gamma\vdash M:\theta$}
  \RightLabel{Weakening}
  \UnaryInfC{$\Gamma,x:J\cdot\theta'\vdash M:\theta$}
  \DisplayProof\\[1.5ex]
  \AxiomC{$\Gamma,x:J\cdot\theta\vdash M:\theta'$}
  \RightLabel{Abstraction}
  \UnaryInfC{$\Gamma\vdash\lambda x.M:J\cdot\theta\multimap\theta'$}
  \DisplayProof\\[1.5ex]
  \AxiomC{$\Gamma\vdash M:J\cdot\theta\multimap\theta'$}
  \AxiomC{$\Gamma'\vdash N:\theta$}
  \RightLabel{Application}
  \BinaryInfC{$\Gamma,J\cdot\Gamma'\vdash MN:\theta'$}
  \DisplayProof\\[1.5ex]
  \AxiomC{$\Gamma,x:J\cdot\theta,y:K\cdot\theta\vdash
    M:\theta'$}
  \RightLabel{Contraction}
  \UnaryInfC{$\Gamma,x:(J+ K)\cdot\theta\vdash M[x/y]:\theta'$}
  \DisplayProof
\end{center}
In \emph{Weakening} we have the side condition $x\not\in fv(M)$, and in  \emph{Application} we require $\text{dom} (\Gamma) \cap \text{dom} (\Gamma') = \emptyset$.
In the \emph{Application} rule we use the notation
\begin{equation}
J\cdot(x_1:K_1\cdot\theta_1,\ldots,x_n:K_n\cdot \theta_n)
\defeq x_1:(J \times K_1)\cdot\theta_1,\ldots,x_n:(J \times K_n)\cdot\theta_n
\end{equation}

For the sake of simplicity we take operations in the semiring to be resolved \emph{syntactically} within the type system. So types such as $2\cdot A$ and $(1+1)\cdot A$ are taken to be syntactically equal. In the context of type-checking this is reasonable because ring actions are always constants that the type-checker can calculate with. If we were to allow resource variables, i.e.\ some form of resource-based polymorphism (cf.~\cite{DBLP:conf/tlca/LagoH09}) then a new structural rule would be required to handle type congruences induced by the semiring theory:
\begin{center}
  \AxiomC{$\Gamma,x:J\cdot\theta'\vdash M:\theta$}
  \AxiomC{$J=_{\mathcal J}J'$}
  \RightLabel{Semiring}
  \BinaryInfC{$\Gamma,x:J'\cdot\theta'\vdash M:\theta$}
  \DisplayProof
\end{center}
But in our current system this level of formalization is not worth the complication. 

\paragraph{Observation.} This is an affine type system where types are
decorated with resources taken from an arbitrary semiring. The new
rules are resource-oriented versions of contraction and application.
The similarity with BLL~\cite{girard1992bounded} and
SCC~\cite{DBLP:journals/tcs/GhicaMO06} is clear. If we instantiate $\mathcal J$ to resource
polynomials (and also remove weakening) we obtain BLL. If we instantiate $\mathcal J$ to the
semiring of natural numbers we get SCC. If $\mathcal J=\{0,1,\infty\}$ we obtain a conventional multiplicative affine type system. 
In Sec.~\ref{chap:pipes} we will see a much more complex resource semiring to control timing of execution. 

\subsection{Type inference}\label{sec:gti}
We present a bound-inference algorithm for the abstract system which works by creating a system of constraints to be solved, separately, by an SMT-solver that can handle the equational theory of the resource semiring. 
In the type grammar, for the exponential type $J\cdot\theta\multimap \theta$ we allow $J$ to stand for a concrete element of $\mathcal{J}$ or for a variable in the input program; the bound-inference algorithm will  produce a set of constraints such that every model of those  constraints gives rise to a typing derivation of the program without
  resource variables as variables are instantiated to suitable concrete values. Type judgments have form
$
\constr{\Gamma\vdash M:\theta}{\chi},
$
where $\chi$ is a set of equational constraints in the semiring. We also allow an arbitrary set of constants $\mathbf k:\theta$, which will allow the definition of concrete programming languages based on the type system. We allow each constant $\textsf k$ to introduce arbitrary resource constraints $\chi_{\textsf k}$
\begin{center}\small
  \AxiomC{ }
 \UnaryInfC{$\constr{x:{\textbf 1}\cdot\theta\vdash x:\theta}{true}$}
  \DisplayProof\\[1.5ex]
  \AxiomC{ }
 \UnaryInfC{$\constr{\emptyset\vdash \textsf{k}:\theta}{\chi_{\textsf k}}$}
  \DisplayProof\\[1.5ex]
  \AxiomC{$\constr{\Gamma\vdash M:\theta}{\chi}$}
 \UnaryInfC{$\constr{\Gamma,x:J\cdot\theta'\vdash M:\theta}{\chi}$}
  \DisplayProof\\[1.5ex]
\AxiomC{$\constr{\Gamma,x:J\cdot\theta\vdash M:\theta'}{\chi}$}
 \UnaryInfC{$\constr{\Gamma\vdash \lambda x:\theta. M:J\cdot\theta\multimap\theta'}{\chi}$}
  \DisplayProof\\[1.5ex]
\AxiomC{$\constr{\Gamma,x:J_1\cdot\theta',y:J_2\cdot\theta''\vdash M:\theta}{\chi}$}
  \UnaryInfC{$\constr{\Gamma,x:J\cdot\theta'\vdash
      M[x/y]:\theta}{\chi \cup\{ J= J_1+J_2\}\cup\overline{\theta'=\theta''}}$}
  \DisplayProof
\\[1.5ex]
  \AxiomC{$\constr{\Gamma\vdash M:J\cdot\theta\multimap\theta'}{\chi}$}
  \AxiomC{$\constr{x_1:J_1\cdot\theta_1,\ldots,x_n:J_n\cdot\theta_n\vdash N:\theta''}{\chi'}$}
  \BinaryInfC{$\constr{\Gamma,x_1:J_1'\cdot\theta_1,\ldots,x_n:J_n'\cdot\theta_n\vdash MN:\theta'}{
       \chi\cup\chi'\cup\{ J_k'=J\cdot J_k\mid 1\leq k\leq n\}\cup\overline{\theta=\theta''}}$}
  \DisplayProof
\end{center}
The constraints of shape $\overline{\theta_1=\theta_2}$ are to be interpreted in the obvious way, as the set of pairwise equalities between resource bounds used in the same position in the two types:
\begin{align*}
\overline{\sigma=\sigma}&\stackrel{def}=\emptyset\\
\overline{J_1\cdot\theta_1\multimap\theta_1'=J_2\cdot\theta_2\multimap\theta_2'}&\stackrel{def}=
\{J_1=J_2\}\cup \overline{\theta_1=\theta_2}\cup \overline{\theta_1'=\theta_2'}.
\end{align*}
If $\mathcal M$ is a model, i.e. a function mapping variables to concrete values, by $\Gamma[\mathcal M]$ we write the textual substitution of each variable by its concrete value in a sequent. The following is then true by construction:
\begin{theorem}
If $\constr{\Gamma\vdash M:\theta}\chi$ and $\mathcal M$ is a model of the system of constraints $\chi$ in the semiring $\mathcal J$ then $(\Gamma\vdash M:\theta)[\mathcal M]$.
\end{theorem}

\subsection{Categorical semantics}\label{sec:cf}
We first give an abstract framework suitable for interpreting the
abstract type system of Sec.~\ref{sub:agf}. We  require two
categories. We interpret \emph{computations} in a symmetric monoidal
closed category $(\mathcal G,\otimes,I)$ in which the tensor unit $I$
is a terminal object. Let $\alpha$ be the \emph{associator} and
$\lambda,\rho$ be the right and left \emph{unitors}. We write the
unique morphism into the terminal object as ${!}_A:A\rightarrow
I$. Currying is the isomorphism
\[
\Lambda_{A,B,C}:A\otimes B\rightarrow C\simeq A\rightarrow B\multimap C,
\] 
and the evaluation morphism is $\mathit{eval}_{A,B}:A\otimes (A\multimap B)\rightarrow B$. 

We interpret  \emph{resources} in a  category $\mathcal R$ with two monoidal tensors $(\catplus,0)$ and $(\cattimes,1)$ such that:

\begin{align*}
  &J\cattimes(K\catplus L) \simeq J\cattimes K \catplus J\cattimes L&\text{(r-distributivity)}\\
  &(J\catplus K)\cattimes L \simeq J\cattimes L \catplus K\cattimes L&\text{(l-distributivity)}\\
  &J\cattimes 0 \simeq 0\cattimes J \simeq 0&\text{(zero)}.
\end{align*}
The action of resources on computations is modeled by a functor
$\cdot:\mathcal R\times\mathcal G\rightarrow \mathcal G$ such that the following natural isomorphisms must exist:
\begin{align}
  \delta_{J,K,A}: J\cdot A \otimes K\cdot A&\simeq(J\catplus K)\cdot A\label{eq:dis2}\\
  \pi_{R,R',A}: R\cdot(R'\cdot A)&\simeq (R\odot R')\cdot A \label{eq:sro}\\
  \zeta_A:0\cdot A &\simeq I\label{eq:zeroi}\\
  \iota_A:\mathbf 1 \cdot A &\simeq A\label{eq:one}
\end{align}
and the following diagrams commute:
\begin{equation}\label{eq:coh}
\xymatrix@C=12ex{
J{\cdot} A \otimes K{\cdot} A \otimes L{\cdot} A \ar[d]^{1_{J{\cdot}A}\otimes \delta_{K,L,A}} \ar[r]^-{\delta_{J,K,A}\otimes 1_{L{\cdot} A}} & (J\catplus K){\cdot} A\otimes L{\cdot} A \ar[d]^{\delta_{J\catplus K,L,A}}\\
J{\cdot} A \otimes (K\catplus L){\cdot} A \ar[r]^-{\delta_{J,K\catplus L, A}} & (J \catplus K\catplus L){\cdot} A 
}
\end{equation}
\begin{equation}\label{eq:nat}
\xymatrix{
J{\cdot}A\otimes K{\cdot} A \ar[d]^{J{\cdot}f\otimes K{\cdot}f} \ar[r]^{\delta_{J,K,A}}& (J\catplus K){\cdot} A\ar[d]^{(J\catplus K)\cdot f} \\
J{\cdot}B\otimes K{\cdot} B\ar[r]^{\delta_{J,K,B}} & (J\catplus K){\cdot} B 
}
\end{equation}
Natural isomorphism $\pi$ (Eqn.~\ref{eq:sro}) reduces successive resource actions on computations to a composite resource action, corresponding to the product of the semiring. 
Natural isomorphism $\delta_{J,K,A}$ in Eqn.~\ref{eq:dis2} is a ``quantitative'' version of the diagonal morphism in a Cartesian category, which collects the resources of the contracted objects. The commuting diagram in Eqn.~\ref{eq:coh} stipulates that the order in which we use the ``quantitative'' diagonal order to contract several objects is irrelevant, and the commuting diagram in Eqn.~\ref{eq:nat} gives a ``quantitative'' counterpart for the naturality of the diagonal morphism. Finally,  Eqns.~\ref{eq:zeroi} and~\ref{eq:one} shows the connection between the units of the tensors involved.

A direct consequence of the naturality of $\rho$ and $I$ being terminal, useful for proving coherence, is:
\begin{proposition}\label{prop:wk}
The following diagram commutes in the category $\mathcal G$ for any $f:B\rightarrow C$:
\[
\xymatrix{
B\ar[r]^-{1_B\otimes !_A}\otimes A \ar[d]^{f\otimes 1_A} & B\otimes I \ar[r]^{\rho_B}& B\ar[d]^f\\
C\otimes A \ar[r]^-{1_C\otimes !_A} & C\otimes I \ar[r]^{\rho_C} & C.
}
\]
\end{proposition}

Computations are interpreted in a canonical way in the category $\mathcal G$. Types are interpreted as objects and terms as morphisms, with 
\begin{align*}
\sbr{J\cdot\theta\multimap\theta'}_{\mathcal G}= (\sbr J_{\mathcal R}\cdot\sbr\theta_{\mathcal G})\multimap\sbr{\theta'}_{\mathcal G}.
\end{align*}
From now on,  the interpretation of the resource action is written as $J$ instead of  $\sbr J_{\mathcal R}$ when there is no ambiguity and the subscript of $\sbr-_{\mathcal G}$ is left implicit. 

Environments are interpreted as
\[
\sbr\Gamma=\sbr{x_1:J_1\cdot\theta_1,\ldots x_n:J_n\cdot \theta_n}
  = J_1\cdot\sbr{\theta_1}\otimes\cdots\otimes J_n\cdot\sbr{\theta_n}.
\]
Terms are morphisms in $\mathcal G$, $\sbr{\Gamma\vdash M:\theta}$ defined as follows:
\begin{align*}
  & \sbr{x:\mathbf 1\cdot\theta\vdash x:\theta} = \iota_{\sbr\theta}  \\
  & \sbr{\Gamma,x:J\cdot\theta \vdash M:\theta'} = 1_{\sbr\Gamma}\otimes !_{J\cdot\sbr\theta};\rho_{\sbr\Gamma};\sbr{\Gamma\vdash M:\theta}   \\
  & \sbr{\Gamma\vdash\lambda x.M:J\cdot\theta\multimap\theta'}
  = \Lambda_{J\cdot\sbr\theta}\bigl( \sbr{\Gamma, x:J\cdot\theta\vdash M:\theta'} \bigr)\\
  & \sbr{\Gamma,J\cdot\Gamma'\vdash FM:\theta'}=(\sbr{\Gamma\vdash F:J\cdot\theta\multimap\theta'}\otimes
  J\cdot\sbr{\Gamma'\vdash M:\theta});\mathit{eval}_{J\cdot\sbr\theta,\sbr{\theta'}}\\
  & \sbr{\Gamma,x:(J+ K)\cdot\theta\vdash M[x/y]:\theta'}=  1_{\sbr\Gamma}\otimes\delta_{J,K, \theta};\sbr{\Gamma,x:J\cdot\theta,y: K\cdot\theta\vdash M:\theta}.
\end{align*}
\subsubsection{Coherence}
The main result of this section is the coherence of  typing. The derivation trees are not unique because there is choice in the use of the weakening and contraction rules. Since meaning is calculated on a particular derivation tree we need to show that it is independent of it. The coherence conditions for the monoidal category are the standard ones~\cite{Kelly64}, but what is interesting and important for coherence is that resource manipulation does not break coherence. The key role is played by the isomorphism $\delta$ which is the resource-sensitive version of contraction, which can combine or de-compose resources without loss of information.  

The key idea of the proof is that we can bring any derivation tree to a standard form (which we call \emph{stratified}), with weakening and contraction performed as late as possible.  
Weakening and contraction for a variable can be pushed as far down as the first lambda abstraction that uses the variable, or to the root of the derivation tree if it remains unbound. 

We will use the following obviously admissible derivation rules (a chain of contractions followed by an abstraction, and weakening followed by an abstraction, respectively):
\begin{center}
\AxiomC{$x_1:J_1\cdot\theta,\ldots,x_n:J_n\cdot\theta,\Gamma\vdash M:\theta'$}
\RightLabel{Abs-con}
\UnaryInfC{$\Gamma\vdash\lambda x.M[x/x_i]:(J_1+\cdots+ J_n)\cdot\theta\multimap\theta'$}
\DisplayProof\\[1.5ex]
\AxiomC{$\Gamma\vdash M:\theta'$}
\RightLabel{Abs-weak}
\UnaryInfC{$\Gamma\vdash\lambda x.M:J\cdot\theta\multimap\theta'$}
\DisplayProof
\end{center}
where, in both rules, $x\not\in fv(M)$.
We also introduce obviously admissible rules for contracting multiple (0, one or more) variables (labeled \emph{Contraction+}) and for weakening multiple (0, one or more) variables (labeled \emph{Weakening+}).

We denote sequents $\Gamma\vdash M:\theta$ by $\Sigma$ and derivation
trees by $\nabla$. Let 
\[\Lambda(\Sigma)\in\{id, wk, ab, ap, co, abco, abwk, co{+}, wk{+}\}\]
be a label on the sequents, indicating whether a sequent is the
product of the rule for identity, weakening, etc. If a sequent
$\Sigma=\Gamma\vdash M:\theta$ is the root of a derivation tree
$\nabla$ we write it $\Sigma^\nabla$ or $\Gamma\vdash^\nabla
M:\theta$. 

We say that a sequent is \emph{linear} if each variable in the environment $\Gamma$ occurs freely in the term $M$ exactly once. 
\begin{definition}
We say that a  derivation tree $\nabla$ is \emph{stratified} if and only if:
\begin{itemize}
\item the root is labeled $wk{+}$;
\item the node above the root is labeled $co{+}$;
\item no other node is labeled by $wk, co, wk{+}$ or $co{+}$;
\item all sequents in $\nabla$, except possibly for the root and the sequent above the root, are linear. 
\end{itemize}
\end{definition}
\begin{lemma}\label{lem:struni}
If a linear sequent has a stratified derivation tree then it is unique (up to renaming of variables). 
\end{lemma}
\begin{proof}
The last two rules (\emph{wk+} and \emph{co+}) bring the sequent to a linear form. 
In constructing the stratified derivation tree $\nabla$ of a linear sequent $\Gamma\vdash M:\theta$ the choice of what rules to apply is always uniquely determined by the structure of the term $M$. 
\begin{description}
\item{$MN$:} The only possible rule is \emph{Application} and, since the term $MN$ is linear both $M$ and $N$ are linear and there is only one way $\Gamma$ can be split. 
\item{$\lambda x.M$:} We consider two cases:
\begin{itemize}
\item If $x\not\in fv(M)$ we infer the rule \emph{Abs-weak}. 
\item If $x\in fv(M)$ we use \emph{Abs-con} to give each occurrence of $x$ in $M$ a new (fresh) name.  
\end{itemize}
There are no other rules that would keep the derivation tree stratified. 
\item{$x$:} The only possible rule is $wk{+}$. 
\end{description}
All the choices above are unique (up to the choice of variable names in \emph{Abs-con}).
\end{proof}
We now show that any  derivation can be reduced to a stratified derivation through applying a series of meaning-preserving tree transformations, which we call \emph{stratifying rules}.

The Weakening rule commutes trivially with all other rules except Identity and Abstraction, if abstraction is on the weakened variable. In this case we replace the sequence of Weakening followed by Abstraction with the combined \emph{Abs-weak} rule.
The more interesting tree transformation rules are for Contraction. 

Contraction commutes with Application. There are two pairs of such rules, one for pushing down contraction in the function and one for pushing down contraction in the argument:
\begin{center}\small
\AxiomC{$ \Gamma,x:J\cdot\theta,y:J'\cdot\theta\vdash F:J_1\cdot\theta_1\multimap \theta_2$}
\UnaryInfC{$\Gamma,x:(J+ J')\cdot\theta\vdash F[x/y]:J_1\cdot\theta_1\multimap \theta_2 $}
\AxiomC{$ \Gamma'\vdash M:\theta_1$}
\BinaryInfC{$ \Gamma,x:(J+ J')\cdot\theta,J_1\cdot\Gamma'\vdash F[x/y]M:\theta_2$}
\DisplayProof
\\[1.5ex]
$\stackrel{AL}\Longleftrightarrow$
\\[1.5ex]
\AxiomC{$\Gamma,x:J\cdot\theta,y:J'\cdot\theta\vdash F :J_1\cdot\theta_1\multimap \theta_2$}
\AxiomC{$\Gamma'\vdash M :\theta_1$}
\BinaryInfC{$\Gamma,x:J\cdot\theta,y:J'\cdot\theta,J_1\cdot\Gamma'\vdash FM :\theta_2$}
\UnaryInfC{$ \Gamma,x:(J+ J')\cdot\theta,J_1\cdot\Gamma'\vdash (FM)[x/y]:\theta_2$}
\DisplayProof
\end{center}
Similarly for pushing down contraction from the argument side and similarly for rules involving weakening:
\begin{center}\small
\AxiomC{$\Gamma\vdash F:J_1\cdot\theta_1\multimap \theta_2 $}
\AxiomC{$ \Gamma,x:J\cdot\theta,y:J'\cdot\theta\vdash M:\theta_1$}
\UnaryInfC{$\Gamma,x:(J+ J')\cdot\theta\vdash M[x/y]:\theta_1$}
\BinaryInfC{$ \Gamma,x:(J_1\times(J+ J'))\cdot\theta,\Gamma'\vdash F(M[x/y]):\theta_2$}
\DisplayProof
\\[1.5ex]
$\stackrel{AR}\Longleftrightarrow$
\\[1.5ex]
\AxiomC{$\Gamma\vdash F :J_1\cdot\theta_1\multimap \theta_2$}
\AxiomC{$\Gamma',x:J\cdot\theta,y:J'\cdot\theta\vdash M :\theta_1$}
\BinaryInfC{$\Gamma,J_1\cdot\Gamma',x:(J_1\times J)\cdot\theta,y:(J_1\times J')\cdot\theta\vdash FM :\theta_2$}
\UnaryInfC{$ \Gamma,x:(J_1\times J+ J_1\times J')\cdot\theta,\Gamma'\vdash (FM)[x/y]:\theta_2$}
\DisplayProof
\end{center}
Contraction also commutes with Abstraction, if the contracted and abstracted variables are distinct, $x\neq y$: 
\begin{center}
\AxiomC{$\Gamma,x:J\cdot\theta,x':J'\cdot \theta,y:K\cdot\theta'\vdash M:\theta''$}
\UnaryInfC{$\Gamma,x:(J+ J')\cdot\theta,y:K\cdot\theta'\vdash M[x/x']:\theta''$}
\UnaryInfC{$\Gamma,x:(J+ J')\cdot\theta\vdash \lambda y.M[x/x']:K\cdot\theta'\multimap\theta''$}
\DisplayProof
\\[1.5ex]
$\stackrel{CA}\Longleftrightarrow$
\\[1.5ex]
\AxiomC{$\Gamma,x:J\cdot\theta,x':J'\cdot \theta,y:K\cdot\theta'\vdash M:\theta''$}
\UnaryInfC{$\Gamma,x:J,x': J'\cdot\theta\vdash \lambda y.M:K\cdot\theta'\multimap\theta''$}
\UnaryInfC{$\Gamma,x:(J+ J')\cdot\theta\vdash (\lambda y.M)[x/x']:K\cdot\theta'\multimap\theta''$}
\DisplayProof
\end{center}
The rule for swapping contraction and weakening is (types are obvious and we elide them for concision):
\begin{center}
\mbox{
  \AxiomC{$\Gamma,y,z\vdash M$}
  \UnaryInfC{$\Gamma,y\vdash M[y/z]$}
  \UnaryInfC{$\Gamma,y,x\vdash M[y/z]$}
  \DisplayProof}$\stackrel{WC}\Longleftrightarrow$\mbox{
\AxiomC{$\Gamma,y,z\vdash M$}
  \UnaryInfC{$\Gamma,y,z,x\vdash M $}
  \UnaryInfC{$\Gamma,y,x\vdash M[y/z]$}
  \DisplayProof}
\end{center}
\begin{proposition}\label{prop:syneq}The following judgments are syntactically equal
\begin{align*}
&\Gamma,x:\theta,\Gamma'\vdash F[x/y]M:\theta'= \Gamma,x:\theta,\Gamma'\vdash (FM)[x/y]:\theta',
\\
&\Gamma,x:(J_1\times(J+ J'))\cdot\theta,\Gamma'\vdash F(M[x/y]):\theta_2=\Gamma,x:(J_1\times J+ J_1\times J')\cdot\theta,\Gamma'\vdash (FM)[x/y]:\theta_2,
\\
&\Gamma,x:(J+ J')\cdot\theta\vdash \lambda y.M[x/x']:K\cdot\theta'\multimap\theta'=\Gamma,x:(J+ J')\cdot\theta\vdash (\lambda y.M)[x/x']:K\cdot\theta'\multimap\theta''.
\end{align*}
\end{proposition}
\begin{proof}
The proof of the first two statements is similar.
Because Application is linear it means that an identifier $y$ occurs either in $F$ or in $M$, but not in both. Therefore $(FM)[x/y]$ is either  $F(M[x/y])$ or $(F[x/y])M$. This makes the terms syntactically equal. In any semiring, $J_1\times(J+ J')=J_1\times J+ J_1\times J'$, which makes the environments equal.  Note that semiring equations are resolved syntactically in the type system, as pointed out at the beginning of this section.  For the third statement we know that $x\neq y$.
\end{proof}
\begin{proposition}\label{prop:str}
If $\nabla$ is a  derivation and $\nabla'$ is a tree obtained by applying a stratifying rule then $\nabla'$ is a valid  derivation with the same root $\Sigma^\nabla=\Sigma^{\nabla'}$ and the same leaves. 
\end{proposition}
\begin{proof}
By inspecting the rules and using Prop.~\ref{prop:syneq}. 
\end{proof}
Stratifying transformations preserve meaning. The following more general proposition shows that in general the weakening rule can be pushed by any other rule without changing meaning.
\begin{lemma} \label{lem:streq}
If $\nabla\Rightarrow\nabla'$ is a stratifying rule then $\sbr{\Sigma^\nabla}=\sbr{\Sigma^{\nabla'}}$.
\end{lemma}
\begin{proof}
By inspecting the rules. Prop.~\ref{prop:str} states that the root
sequents are equal and the trees are well-formed. For WC (and the other rules involving the
stratification of Weakening) this is an immediate consequence of
Prop.~\ref{prop:wk}. For AL and AR the equality of the two sides is an
immediate consequence of symmetry in $\mathcal G$ and the
functoriality of the tensor $\otimes$. For CA the equality of the two sides is an instance of
the general property in a symmetric monoidal closed category that $f;\Lambda(g)=\Lambda((f\otimes 1_{B'});g)$ for 
any $A\stackrel f\rightarrow B$, $B\otimes B'\stackrel g\rightarrow C$. 
\end{proof}
\begin{lemma}\label{lem:twk}
If $\nabla,\nabla'$ are   derivation trees consisting only of \emph{Contraction} and \emph{Weakening} with
a common root $\Sigma$ then
$\sbr{\Sigma^\nabla}=\sbr{\Sigma^{\nabla'}}$. 
\end{lemma}
\begin{proof}
Weakening commutes with any other rule (Prop.~\ref{prop:wk}).
Changing the order of multiple contraction of the same variable
  uses the associativity coherence property in Eqn.~\ref{eq:coh}.
Changing the order in which different variables are contracted
  uses the naturality coherence property in Eqn.~\ref{eq:nat}.
\end{proof}
\begin{lemma}\label{lem:streq2}
If $\nabla$ is a  derivation there exists a stratified derivation
tree $\nabla'$ which can be obtained from $\nabla$ by applying a
(finite) sequence of stratifying tree transformations. Moreover, $\sbr{\Sigma^\nabla}=\sbr{\Sigma^{\nabla'}}$.
\end{lemma}
\begin{proof}
The stratifying transformations push contraction and weakening through
any other rules and the derivation trees have finite height. If a contraction or weakening
cannot be pushed through a rule it means that the rule is an abstraction on the 
variable being contracted or weakened, and we replace the two rules with 
either \emph{Abs-con} or \emph{Abs-weak}. 

For the weakening and contractions pushed to the bottom of the tree the order
is irrelevant, according to Lem.~\ref{lem:twk}, therefore we replace them with
a \emph{Contraction+} and \emph{Weakening+} which perform all the required weakening
and contraction in one step each.
The result is a stratified  tree.

Then we
apply induction on the chain of stratifying rules using
Lem.~\ref{lem:streq} for every rule application and 
Lem.~\ref{lem:twk} for the final chain of weakening
and contractions.
\end{proof}
\begin{theorem}[Coherence]
For any   derivation trees ${\nabla_1},{\nabla_2}$ with common root $\Sigma$,  $\sbr{\Sigma^{\nabla_1}}=\sbr{\Sigma^{\nabla_2}}$.
\end{theorem}
\begin{proof}
  Using Lem.~\ref{lem:streq}, $\nabla_1,\nabla_2$ must be effectively stratifiable into
   trees $\nabla_1',\nabla_2'$ with the same root. Using Lem.~\ref{lem:streq2},
  $\sbr{\Sigma^{\nabla_i}}=\sbr{\Sigma^{\nabla_i'}}$ for
  $i=1,2$. We first reduce $\Sigma^{\nabla_i}$ to a linear form using \emph{Contraction+} and \emph{Weakening+} then use Lem.~\ref{lem:struni},
  ${\Sigma^{\nabla_1'}}={\Sigma^{\nabla_2'}}$. 
\end{proof}

\section{Case study: automated pipelining}\label{chap:pipes}
Let us instantiate the abstract type system to a non-trivial resource-sensitive type system: automatic pipelining of computations.  This is interesting for two reasons. First we get to work with a complex resource semiring of \emph{execution schedules}. Second, for the type inference we show how the intrinsic constraints system generated over the resource semiring can be seamlessly combined with additional extrinsic constraints, in our case imposing a pipelining (first-in-first-out) discipline on the schedules. 

The concrete type system is an instance of the generic type system when $\mathcal J$ is taken to be the semigroup semiring (i.e.\ multisets) of one-dimensional contractive affine\footnote{The word ``affine'' has two distinct technical meanings, both standard: logical vs. algebraic. The overloading should be unambiguous in context.} transformations
$
\mathcal J = \mathbb N[\aff].
$
We will use the notation $J=[x_1,x_2,\ldots,x_n]$ to represent some $J$ as a multiset; we call $x_i$ its \emph{stages} and $J$ a \emph{schedule}.

Contractive affine transformations enable composition of timed functions in a natural way. Our view of timing is \emph{relative}: in a type $([x_1,\ldots,x_n]\cdot A)\multimap B$ (brackets added for emphasis) we take the execution of the function to always be, by convention, the unit interval. This is a call-by-name language so each argument is re-evaluated when needed (to prevent needless re-evaluation we can use the store explicitly). The size of the multiset indicates that the function uses its argument $n$ times. Contractive affine transformation $x_i$, when applied to the unit interval, yields a sub-interval indicating the timing of execution of the $i$-th use of the argument. Compositionality is given automatically by the fact that the product of contractive affine transformations is a a contractive affine transformation. Composing time represented as explicit intervals can be done but is more complicated. 

A \emph{contractive} affine transformation is represented
$
x = 
\left(
\begin{matrix}
s & p \\
0 & 1
\end{matrix}
\right)
\in\aff$, where  $0\leq s\leq 1$ and $0\leq s+ p\leq 1.
$
The factor $s$ is a \emph{scaling factor}, representing the \emph{relative duration}
of a computation, and $p$ is a \emph{phase}, representing a \emph{relative delay} for the same  computation. A one-dimensional affine transformation acting on the unit interval, in affine representation, can be used to represent the duration of the computation of one run of a term starting at $t_0$ and ending at $t_1$:
\[
\left(
\begin{matrix}
s & p \\
0 & 1
\end{matrix}
\right)
\times
\left(
\begin{matrix}
0 & 1 \\
1 & 1
\end{matrix}
\right)
=\left(
\begin{matrix}
p & s+p \\
1 & 1
\end{matrix}
\right)\defeq
\left(
\begin{matrix}
t_0 & t_1 \\
1 & 1
\end{matrix}
\right)
\]

\begin{proposition}
If $x,y\in\aff$ then $x\times y\in\aff$. 
\end{proposition}

When we refer to the timing of a computation, and it is unambiguous from context, we will sometimes use just $x$ to refer to its action on the unit interval $u=[0,1]$. For example, if we write $x\subseteq x'$ we mean $x\cdot u\subseteq x'\cdot u$, i.e.\   $[p,s+p]\subseteq[p',s'+p']$, i.e. $p\geq p'$ and $s+p\leq s'+p'$.  If we write $x\leq x'$ we mean the Egli-Milner order on the two intervals, $x\cdot u\leq x'\cdot u$, i.e.\   $p\leq p'$ and $s+p\leq s'+p'$.  If we write $x\cap x'=\emptyset$ we mean the two intervals are disjoint, $x\cdot u\cap x'\cdot u=\emptyset$, etc.

Contractive affine transformations form a semigroup with matrix product as multiplication and unit element
$
I\defeq
\left(
\begin{matrix}
1 & 0 \\
0 & 1
\end{matrix}
\right)
$.
The semiring of a semigroup $(\mathcal G,\times,I)$ is  a natural construction from any semiring and any semigroup. In our case the semiring is natural numbers ($\mathbb N$), so the semigroup semiring is the set of finitely supported functions $J:\aff\rightarrow \mathbb N$ with 
\begin{align}
	\ringzero (x) &= 0\label{eq:zero}\\
	\ringunit (x) &= \begin{cases}
	1 & \text{if } x = I\\
	0 & \text{otherwise}
	\end{cases}\label{eq:unit}\\
	(J+ K)(x) &= J(x)+K(x)\label{eq:add}\\
	(J\times K)(x) &= \sum_{\substack{y,z\in\aff\\y\times z=x}} J(y)\times K(z).\label{eq:mul}
\end{align}
This is isomorphic to finite multisets over $\aff$. We use interchangeably whichever representation is more convenient. 

\subsection{A concrete programming language}

A concrete programming language is obtained by adding a family of functional constants in the style of Idealized Algol~\cite{reynolds1997essence}. Let us call it PIA (Pipelined-IA). We take commands and integer expressions as the base types, 
$
\sigma::=\lcomt\mid\lexpt.
$

Ground-type operators are provided with explicit timing information. For example, for commands we have a family of timed composition operators (i.e.\ schedulers):
\[
\lcompc_{x,y}:[x]\cdot\lcomt\multimap [y]\cdot\lcomt\multimap\lcomt.
\]
The fact that $x,y$ are contractive is a \emph{causality} constraint which says that each argument must execute within the interval in which the main body of the function is running which is, by convention, the unit interval. Sequential composition is a scheduler in which the arguments are non-overlapping, with the first argument completing before the second argument starts:
$\mathsf{seq}_{x,y} = \lcompc_{x,y}$ where $  x\leq y$ and $x\cap y=\emptyset$ (which we write $x<y$). Parallel composition is simply
$
\mathsf{par}_x = \lcompc_{x,x},
$
with both arguments initiating and completing execution at the same time. Schedulers that are neither purely sequential nor parallel, but a combination thereof, are also possible. 

Arithmetic operators are also given explicit timings, but branching needs to be sequential.
\begin{align*}
\lopc_{x,y}&:[x]\cdot\lexpt\multimap [y]\cdot\lexpt\multimap\lexpt,\\
\lifc_{x,y} &: [x]\cdot\lexpt\multimap [y]\cdot\sigma\multimap [y]\cdot\sigma\multimap\sigma,
	\quad   x<  y.
\end{align*}

Assignable variables are handled by separating read and write access, as is common for IA. Let the type of \emph{acceptors} be defined (syntactically) as $\lacct\defeq[w]\cdot\lexpt\multimap\lcomt$, where $w\in\aff$ is a system-dependent constant (writing to memory cannot usually be instantaneous). There is no stand-alone $\lvart$ type in PIA, instead the readers and writers to a variable are bound to the same memory location by a block variable constructor with signature:
\[
\lnewc_{\sigma,J,K}:(J\cdot\lexpt\multimap K\cdot\lacct\multimap \sigma)\multimap \sigma, \quad \sigma\in\{\lexpt, \lcomt\}.
\]
For programmer convenience $\lvart$-typed identifiers can be sugared into the language but, because the read and write schedules of access need to be maintained separately, the contraction rules become complicated (yet routine) so we omit them here. 

Finally, ground-type constants are 
$
1:\lexpt$ and $\lskipc:\lcomt.
$

In order to keep execution deterministic and timing predictable, no constants with data-dependent timing of execution can be allowed, such as recursion, iteration or semaphores. These restrictions are not onerous. Unbounded recursive (or iterated) executions cannot be in general pipelined, only finite unfoldings; we could handle this but it is a conceptually uninteresting complication. Semaphores are asynchronous computational features that also involve non-deterministic waiting for conditions to happen and cannot be pipelined. The language presented here must be understood as a sub-language of a larger ambient programming language, defining those computations that can be pipelined. 

\begin{example} \label{ex:incx}
The program $\lnewc\,x.\,x:={!}x+1$, written functionally (while separating the reader and the acceptor) as $\lnewc(\lambda x_r\lambda x_w.x_w(\mathsf{add}\,x_r\,1))$ is typable. We give one possible way to annotate the constants with timing so that the term types:
$
\lnewc_{\lcomt,[w\times b],[w]}(\lambda x_r\lambda x_w.x_w(\mathsf{add}_{b,b}\,x_r\,1))
$
which, written in a fully sugared notation, would be:
$
\lnewc_{\lcomt,[w\times b],[w]}\,x:={!}x+_{b,b}1.
$
Note that addition here is given the schedule of a parallel operation with some arbitrary schedule $b$.
\end{example}

\subsection{Type inference for automated pipelining}\label{sec:typin}\label{sec:pip}

Note that the recipe from Sec.~\ref{sec:gti} cannot be immediately applied because there is no (off-the-shelf) SMT solver for $\mathbb N[\text{Aff}_1^c]$. We need to run the SMT in two stages: first we calculate the sizes of the multiset (as in SCC inference), which allows us to reduce constraints in $\mathbb N[\text{Aff}_1^c]$ to constraints in $\text{Aff}_1^c$. Then we map equations over $\text{Aff}_1^c$ into real-number equations, which can be handled by the SMT solver. There is a final, bureaucratic, step of reconstructing the multi-sets from the real-number values. To fully automate the process we also use Hindley-Milner type inference to determine the underlying simple-type structure~\cite{milner1978theory}. 

Multiset size (SCC) type inference is presented in detail elsewhere~\cite{DBLP:conf/popl/GhicaS11}, but we will quickly review it here in the context of PIA. We first interpret schedules as natural numbers, representing their number of stages $J\in\mathbb N$. Unknown schedules are variables, schedules with unknown stages but fixed size (such as those for operators) are constants. A type derivation results in a constraint system over $\mathbb N$ which can be solved by an SMT tool such as Z3~\cite{Z3}. More precisely, Z3 can attempt to solve the system, but it can be either unsatisfiable in some cases or unsolvable as nonlinear systems of constraints over $\mathbb N$ are generally undecidable. 

As a practical observation, solving this constraint using general-purpose tools will give an arbitrary solution, if it exists, whereas a ``small'' solution is preferable. In~\cite{DBLP:conf/popl/GhicaS11} we give a special-purpose algorithm guaranteed to produce solutions that are in a certain sense minimal. To achieve a small solution when using Z3 we set a global maximum bound which we increment on iterated calls to Z3 until the system is satisfied. 

The next stage is to instantiate the schedules to their known sizes, and to re-run the inference algorithm, this time in order to compute the stages. This stage proceeds according to the general type-inference recipe, resulting in a system of constraints over the $\mathbb N[\text{Aff}_1^c]$ semiring, with the particular feature that all the sizes of all the multisets is known. We only need to specify the schedules for the constants:
\begin{center}
  \AxiomC{ }
 \UnaryInfC{$\constr{\emptyset\vdash 1: \lexpt}{\true}$}
  \DisplayProof\\[1.5ex]
  \AxiomC{ }
 \UnaryInfC{$\constr{\emptyset\vdash \lskipc: \lcomt}{\true}$}
  \DisplayProof\\[1.5ex]
  \AxiomC{ }
 \UnaryInfC{$\constr{\emptyset\vdash \lopc_{x,y}:[x]\cdot\sigma\multimap[y]\cdot\sigma\multimap\sigma}{\{x\neq I, y\neq I\}}$}
  \DisplayProof\\[1.5ex]
  \AxiomC{}
  \UnaryInfC{$\constr{\emptyset\vdash\lifc_{x,y}:[x]\cdot\lexpt\multimap[y]\cdot\sigma\multimap[y]\cdot\sigma\multimap\sigma}{x< y}$}
  \DisplayProof\\[1.5ex]
  \AxiomC{}
 \UnaryInfC{$\constr{\emptyset{\vdash} \lnewc_{\sigma,J,K}:(J{\cdot}\lexpt\multimap K{\cdot}\lacct\multimap\sigma)\multimap\sigma}{\{0\not\in K\}}$}
  \DisplayProof
\end{center}
 In the concrete system it is useful to characterize the resource usage of families of constants also by using constraints, which can be simply combined with the constraints (in the theory of the semiring) produced by the generic type inference algorithm. The language of constraints itself can be extended arbitrarily, provided that eventually we can represent it into the language of our external SMT solver, Z3. The constraints introduced by the language constants are motivated as follows:
\begin{description}
\item[op:] We prevent the execution of any of the two arguments to take the full interval, because an arithmetic operation cannot be computed instantaneously.
\item[if:] The execution of the guard must precede that of the branches.
\item[new:] The write-actions cannot be instantaneous.  
\end{description}

This allows us to translate the constraints into real-number constraints. Solving the system (using Z3) gives precise timing bounds for all types. However, this does not guarantee the fact that computations can be pipelined, it just establishes timings. 
In order to force a pipeline-compatible timing discipline we need to add extra constraints guaranteeing the fact that each timing annotation $J$ is in fact a proper pipeline. Two stages $x_1,x_2\in\aff$ are \emph{FIFO} if they are Egli-Milner-ordered, $x_1\leq x_2$. They are \emph{strictly FIFO}, written $x_1\lhd x_2$ if they are FIFO and they do not start or end at the same time, i.e. if $x_i\cdot [0,1] = [t_i,t_i']$ then $t_0\neq t_0'$ and $t_1\neq t_1'$. 

\newcommand{\pipe}{\mathsf{Pipe}}
\begin{definition}\label{def:pipe}
We say that a schedule $J\in\mathbb N[\aff]$ is a \emph{pipeline}, written $\pipe(J)$, if and only if $\forall x\in \aff, J(x)\leq 1$ (i.e.\ $J$ is a proper set) and for all $x, x'\in J$,  either $x\lhd x'$ or $x'\lhd x$ or $x=x'$. 
\end{definition}

Given a system of constraints $\chi$ over $\mathbb N[\aff]$, before solving it we augment it with the condition that every schedule is a proper pipeline: for any $J$ used in $\chi$, $\pipe(J)$. Using the conventional representation (scaling and phase), the usual matrix operations and the pipelining definitions above we can represent $\chi$ as a system of constraints over $\mathbb R$, and solve it using Z3. 

\paragraph{Implementation note.} For the implementation, we enforce arbitrary orders on the stages of the pipeline and, if that particular order is not satisfiable then a different (arbitrary) order is chosen and the process is repeated. However, spelling out the constraint for the existence of a pipelining order $\lhd$ for any schedule $J$ would entail a disjunction over all possible such orders, which is $\mathcal O(n!)$ in the size of the schedule, for each schedule, therefore not realistic. However, if the systems of constraints have few constants and mostly unknowns, i.e.\ we are trying to find a schedule rather than accommodate complex known schedules, our experience shows that this pragmatic approach is reasonable. 

\paragraph{}
Ex.~\ref{ex:incx} is from a scheduling point of view quite trivial because no pipelining takes
place. We consider two more complex examples below. 

\begin{example}\label{ex:fx4}
Let us first consider the simple problem of using three parallel adders to compute the sum $f x + f x + f x + f x$ when we know the timings of $f$. Suppose $f:([(0.5, 0.1);(0.5, 0.2)]\cdot \lexpt\multimap \lexpt$, i.e.\ it is a two-stage pipeline where the execution of the argument takes half the time of the overall execution and have relative delays of 0.1 and 0.2 respectively. We have the choice of using three adders with distinct schedules $+_i:[x_i]\cdot\lexpt\multimap [y_i]\cdot\lexpt\multimap\lexpt$ ($i\in\{1,2,3\}$) so that the expression respects the pipelined schedule of execution of $f$. The way the operators are associated is relevant: $(f x +_2 f x) +_1 (f x +_3 f x)$. Also note that part of the specification of the problem entails that the adders are trivial (single-stage) pipelines. Following the algorithm above, the typing constraints are resolved to the following:
\begin{align*}
+_1 &:[(0.5, 0.265625)]\cdot\lexpt \multimap [(0.5, 0.25)]\cdot\lexpt\multimap\lexpt\\
+_2 &:[(0.5, 0.21875)]\cdot\lexpt \multimap [(0.5, 0.25)]\cdot\lexpt\multimap\lexpt\\
+_3 &:[(0.5, 0.375)]\cdot\lexpt \multimap [(0.5, 0.25)]\cdot\lexpt\multimap\lexpt
\end{align*}
In the implementation, the system of constraints has 142 variables and 357 assertions, and is solved by Z3 in circa 0.1 seconds on a high-end desktop machine.
\end{example}

\begin{example}\label{ex:fx4}
Let us now consider a more complex, higher-order example. Suppose we want to calculate the convolution ($*$) of a pipelined function ($f:[(0.5, 0.1);(0.5, 0.2)]\cdot\lexpt\multimap\lexpt$) with itself four times. And also suppose that we want to use just two instances of the convolution operator $*_1, *_2$, so we need to perform contraction on it as well. The simple type of the convolution operator is
$
(*):(\lexpt\rightarrow\lexpt)\rightarrow(\lexpt\rightarrow\lexpt)\rightarrow \lexpt\rightarrow\lexpt.
$ For hardware compilation this corresponds to the following circuit diagram:
\begin{center}
\includegraphics{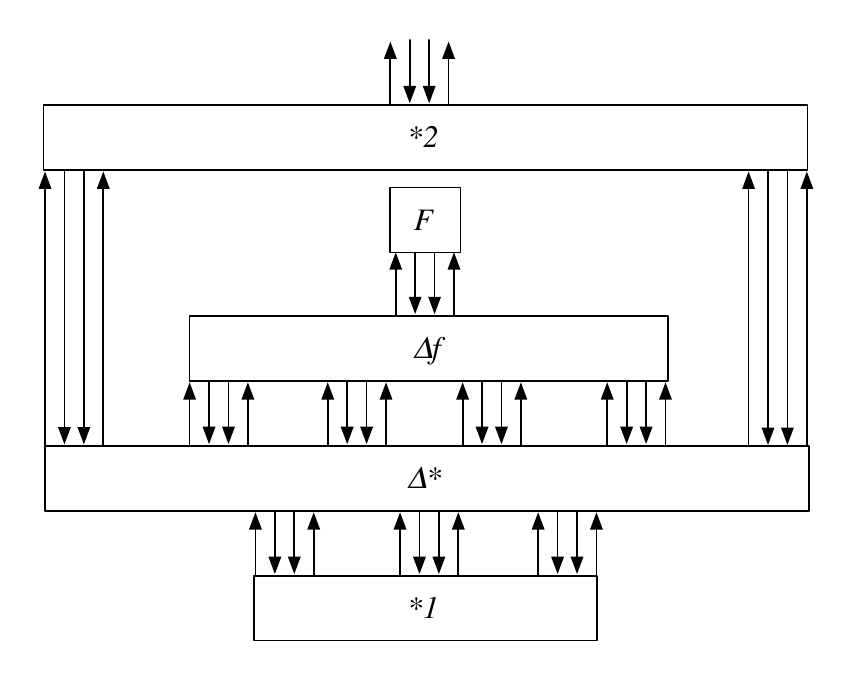}
\end{center}
By $F$ we denote the circuit implementing the function $f$, $*1, *2$ the two instances of the convolution operator, $\Delta f$ the four-way contraction of $f$ and $\Delta *$ the contraction of the convolution operation itself. Every port in this diagram must observe a pipelining discipline. 

The implementation of $f$ and $*$ are unknown, so we want to compute the timings for the term
\begin{align*}
(*_1)&:J_1^{vi}\cdot(J_1^i\cdot(J_1^{ii}\cdot\lexpt\multimap\lexpt)\rightarrow J_1^{iv}\cdot(J_1^{iii}\cdot\lexpt\multimap\lexpt)\multimap J_1^v\cdot\lexpt\multimap\lexpt),\\
(*_2)&:J_2^{vi}\cdot(J_2^i\cdot(J_2^{ii}\cdot\lexpt\multimap\lexpt)\rightarrow J_2^{iv}\cdot(J_2^{iii}\cdot\lexpt\multimap\lexpt)\multimap J_2^v\cdot\lexpt\multimap\lexpt),\\
f&:J_3\cdot([(0.5, 0.1);(0.5, 0.2)]\cdot\lexpt\multimap\lexpt)\vdash
(f *_1 f) *_2 (f *_1 f): \theta.
\end{align*}
The constraint system has 114 variables and 548 assertions and is solved by Z3 in 0.6 seconds on a high-end desktop machine. The results are:
\begin{align*}
J_1^i &=J_1^{iv}=J_2^i =J_2^{iv}=[(1.0, 0.0)]\\
J_1^{ii} &=J_1^{iii}=J_1^v=J_2^{ii} =J_2^{iii}=J_2^v=[(0.5, 0.1);(0.5, 0.2)]\\
J_1^{vi}&=J_3=[(0.5, 0.125);(0.5, 0.25);(0.5, 0.375);(0.5, 0.4375)]\\
J_2^{vi}&=[(0.25, 0.25);(0.25, 0.5);(0.25, 0.625)]
\end{align*}
\end{example}

\section{Timed games: semantics of PIA}\label{sec:gamhort}
Rather than give a conventional operational semantics to our programming language we define it denotationally, using game semantics. This has the technical advantage that the model is compositional by construction. Moreover, game semantics packs pertinent operational intuitions and can be effectively presented, therefore (arguably) not much is lost by eschewing the conventional syntax-oriented operational semantics. In support of this statement we mention our prior work on hardware~\cite{DBLP:conf/mpc/Ghica12} and distributed~\cite{fredrikssong13} compilation directly from the game-semantic model of a programming language. 
Game-semantic models are also well suited to modeling resource usage
explicitly by annotating moves with tokens representing resource
usage~\cite{DBLP:conf/popl/Ghica05}. We will use an annotated game
model here as well, starting from the game model of
ICA~\cite{DBLP:journals/apal/GhicaM08}. 

Note that since we are giving a denotational semantics the usual syntactic sanity checks (reduction preserves typing) do not apply. Instead, we must show that our model fits the categorical requirements of Sec.~\ref{sec:cf}. These requirements subsume and strengthen the syntactic sanity checks by lifting them to higher order terms and formulating them compositionally. 

This section assumes that the reader is familiar with the basic concepts of game
semantics. Tutorial introductions to game semantics are available, e.g.~\cite{DBLP:conf/lics/Ghica09}. 
For readability, all techincal proofs are given in a Sec.~\ref{sec:tec}.

\begin{definition}[Pre-arena]\label{def:prota}
A \emph{pre-arena} $A$ is a tuple $(M,\tau,E,\lambda,\vdash,{\asymp})$:
\begin{enumerate}
\item \label{i:moves}
 $M$ is a set of \emph{moves};
\item \label{i:timing}
$\tau: M\rightarrow [0,1]$ is a  \emph{timing function};
\item \label{i:label}
$\lambda:M\rightarrow \{\mv O,\mv P\}\times\{\mv Q, \mv A\}\times \{\mv M,\mv N\}$ is a \emph{labelling function};
\item \label{i:initial}
$E\subseteq M$ such that $\lambda(E)=(\mv O,\mv Q, x)$ for some $x\in\{\mv M, \mv N \}$;
\item \label{i:enable}
${\vdash}\subseteq M\times M$ is an \emph{enabling relation}, such that for any $n$ there is an $m$ such that $m\vdash n$ if and only if $n\not\in E$, and for any $m, n$, $m\vdash n$ implies 
\begin{enumerate}
\item \label{i:oqenable}
$(\pi_1\circ\lambda)(m)\not= (\pi_1\circ\lambda)(n)$,
\item \label{i:qenable}
$(\pi_2\circ\lambda)(m)=\mv Q$, 
\item \label{i:placehold}
if $(\pi_3\circ\lambda)(m)=\mv N$ then $(\pi_3\circ\lambda)(n)=\mv N$.
\end{enumerate}
\item \label{i:simult}
${\asymp}\subseteq M\times M$ is an equivalence relation such that for any $m, m', n, n'$ 
\begin{enumerate}
\item \label{i:simtim}
if $m\asymp n$ then $\tau(m)=\tau(n), (\pi_i\circ\lambda)(m)=(\pi_i\circ\lambda)(n)$ for $i=1,2$,
\item \label{i:simjus}
if $m\vdash n, m'\vdash n', n\asymp n'$ then $m\asymp m'$,
\item \label{i:simena}
if $m\vdash n, m\vdash n', (\pi_1\circ\lambda)(n)=\mv A, (\pi_1\circ\lambda)(n')=\mv A$ then $n\asymp n'$.
\end{enumerate}
\end{enumerate}
\end{definition}
Some of the game-semantic concepts are conventional (move, opponent-proponent, question-answer, enabling, initial move) but some are specific to timed games. We will use $q, a, o, p$ to stand for question, answer, opponent or proponent move if ambiguities are not introduced. We also use $\overline E$ to signify the set of \emph{final answers}, the answers to the initial questions in $E$.

The key new concept particular to timed systems is that of timing (Def.~\ref{def:prota}.\ref{i:timing}), assigning each move a time in the unit interval. The arenas of timed games introduce the notion of \emph{alternative} moves, moves that are simultaneous (in the arena) but only one of which can occur in an actual play. Alternative moves are related by $\asymp$. Answers to the same question are alternatives (Def.~\ref{def:prota}.\ref{i:simena}) as are a move and its dummy counterpart (Def.~\ref{def:prota}.\ref{i:simtim}). 

One of the $\asymp$-alternatives in a collection of moves is the \emph{dummy move}.
The notion of ``dummy move'' (or \emph{non-move}, the label $\mv N$) in Def.~\ref{def:prota}.\ref{i:label} corresponds to the principle that in a timed system observations are driven by timing: at any given moment we can observe a system to check whether it is producing any output or requesting any input. If that is the case this is modeled by a conventional, actual, move. But if that is not the case, especially if at a given time a move was possible or expected, the fact that no move occurred is relevant, and modeled by a dummy (non)move. If a dummy move enables another move, then that move must also be a dummy (Def.~\ref{def:prota}.\ref{i:placehold}). 

\begin{definition}[Arena]\label{def:arena}
	A \emph{precedence relation for arena} $A=(M,\tau,E,\lambda,\vdash,{\asymp})$, ${\prec_A}\subseteq M\times M$ is the minimum transitive relation such that:
\begin{enumerate}
\item if $m\vdash n$ then $m\prec_A n$;
\item if $\tau(m)<\tau(n)$ then $m\prec_A n$;\label{def:arenat}
\item if $m\asymp m' \prec_A n$ then $m\prec_A n$;
\item if $m\prec_A n'\asymp n$ then $m\prec_A n$;
\item \label{i:fj} if $q\vdash a,q'\vdash a'$ and $q\vdash q'$ then $a'\prec_A a$.
\end{enumerate}
	An \emph{arena} is a pre-arena that has a well-founded precedence relation. 
\end{definition}
Precedence is consistent with timing but it has a finer grain: even moves with the same timing may have a precedence relation, which indicates causality. As in synchronous digital systems, just because two signals have the same timing (are on the same cycle) does not mean they are truly simultaneous, as time itself is only an abstract approximation. Within the same timing we are just unable to further discern the \emph{value} of the time but we can still observe the \emph{order} of the events. This distinction is essential in preventing causal loops in composition. The last requirement (Def.~\ref{def:arena}.\ref{i:fj}) is the language-dependent requirement that all children of a thread terminate before the parent (the \emph{Fork} and \emph{Join} rules in game semantics of ICA).

For any arena $A$, two time intervals will play an important role, the time interval when a play \emph{may} execute $t_M$ and the time when a play \emph{must} execute $t_m$, defined as
\begin{align*}
t_M = [\text{inf}(\tau (E)), \text{sup}(\tau(\overline E))],\quad
t_m = [\text{sup}(\tau (E)), \text{inf}(\tau(\overline E))].
\end{align*}
A play in an arena may (must) execute after the earliest (latest) initial move and before the latest (earliest) final move. 
\begin{definition}
If $x\in\aff$ then the action of $x$ on $A(M,\tau,E,\lambda,\vdash,{\asymp})$ is 
$
x\cdot A = (M,\tau',E,\lambda,\vdash,{\asymp})
$,
where $\tau'(m)=t'$ if $x\times\left(\begin{smallmatrix} \tau(m) \\ 1 \end{smallmatrix} \right)=\left(\begin{smallmatrix} t'\\ 1 \end{smallmatrix} \right)$
\end{definition}
\begin{definition}\label{def:schedact}
If $J\in\mathbb N[\aff]$ is a schedule then the action of the schedule on the arena $A$ is 
$
J\cdot A = \biguplus_{x\in\aff}\biguplus_{n\leq J(x)}x\cdot A.
$
\end{definition}
Every stage in the schedule is allowed to act on the arena $A$. Moreover, if a schedule occurs several times in the schedule then as many copies are created from the arena as required. The notion of ``distinct copies of the same arena" can be formalised using either explicit tags or nominal techniques, but we avoid this formalisation here, whenever possible, as it is generally the case in presentations of game-semantic models, in order to keep the technicalities at bay. 

Let $[-,-]$ be the co-pairing of two functions and $\lambda^\bullet$ be a function like $\lambda$ except that the $\mv O, \mv P$ value are swapped. Let $A\otimes A'$ be the (disjoint) union of two arenas. 
\begin{definition}
If $A=(M,\tau,E,\lambda,\vdash,{\asymp})$ and $A'=(M',\tau',E',\lambda',\vdash',{\asymp'})$ are arenas such that $t_M(A)\subseteq t_m(A')$ then we define the \emph{arrow} arena as $A\multimap A'=(M\uplus M',\tau\uplus \tau',E',[\lambda^\bullet,\lambda'],{\vdash''},{\asymp}\uplus{\asymp'})$ where $\vdash''$ is defined as ${\vdash''}={\vdash}\uplus{\vdash'}\uplus\{(e',e)\mid e\in E, e'\in E', (\pi_3\circ \lambda)(e)=\mv N)\}\uplus\{(e',e)\mid e\in E, e'\in E', (\pi_3\circ \lambda)(e')=\mv M)\}$.
\end{definition}
The arrow arena has the conventional definition in game semantics, except for the \emph{causality} requirement that $t_M(A)\subseteq t_m(A')$ which ensures that all possible computations of the argument happens within the time bounds of the calling arena. Note that this condition is quite restrictive because it does not take into account the enabling relation, just the absolute earliest and latest possible moves in the arenas.
\begin{lemma}
If $J\in\mathbb N[\aff]$ and $A, A'$ are arenas then $J\cdot A$, $A\multimap A'$ and $A\otimes A'$ are arenas. 
\end{lemma}

\begin{definition}[Play]\label{def:play}
A play $P$ on an arena $A$ is a sequence of \emph{distinct} moves of $A$ such that 
\begin{enumerate}
\item\label{it:jus} for any $m\in P$, $m$ is initial or there is a unique $n\in P$ such that $n\vdash m$;
\item\label{it:qa}  for any $q\in P$, there exists a unique $a\in P$ such that $q\vdash a$;
\item\label{it:prec}  for any $m\prec_A m' \in P$, $m$ occurs before $m'$ in $P$;
\item\label{it:sym}  
for any  $m\in M$ there is a unique $n\in P$ such that $m\asymp n$. 
\end{enumerate}
The set of all legal plays of arena $A$ is $\mathcal L(A)$. 
\end{definition}
Some of the rules are common in game semantics, such as the existence of unique enablers (Def.~\ref{def:play}.\ref{it:jus}) and unique answers (Def.~\ref{def:play}.\ref{it:qa}). Clearly, temporal precedence must be consistent with move sequencing in the play (Def.~\ref{def:play}.\ref{it:prec}). The last condition (Def.~\ref{def:play}.\ref{it:sym}) requires that exactly one of a set of alternative moves occurs in a play. 

Also note that plays must consist of distinct moves in the arenas: the enabling relation is a directed acyclic graph and a play is a path in this DAG. This is why timing is associated with arenas rather than with plays.  

A \emph{position} of an arena $A$ is a prefix of some play $P$ in $A$. A move $m$ is \emph{legal} in a position $P$ if $P::m$ is a position. A \emph{position} of a set of plays $\sigma$ is a prefix of some play $P\in\sigma$ in that set. 
\begin{definition}[Strategy]\label{def:strat}
A \emph{strategy} $\sigma$ on arena $A$, written $\sigma:A$, is a set of plays on $A$ such that 
\begin{enumerate}
\item (responsive) for any position $Q$ in $\sigma$ and legal $\mv O$-move $o$ for $Q$ in $A$, $Q::o$ is a position in $\sigma$;
\item (saturated) for any play $P::m::m'::P'\in\sigma$ if $m$ is a $\mv P$-move or $m'$ is a $\mv O$-move (or both) and $P::m'::m::P'$ is a play then it is also in $\sigma$.
\end{enumerate}
\end{definition}
These conditions correspond to O-completeness and saturations, as used in the ICA game model. 

A move that (transitively) enables two moves is said to be a \emph{common enabler}. A common enabler that enables no other common enablers is said to be their \emph{last common enabler}. 
\begin{definition}\label{def:precs}
A \emph{precedence relation for strategy} $\sigma:A$ is a transitive well-founded relation ${\prec_\sigma}\subseteq M_A\times M_A$ such that :
\begin{enumerate}
\item for any $m,m'\in M_A$, if $m\prec_A m'$ then $m\prec_\sigma m'$;
\item if $m\asymp m' \prec_\sigma n$ then $m\prec_\sigma n$;
\item if $m\prec_\sigma n'\asymp n$ then $m\prec_\sigma n$;
\item for any position $Q$ of $\sigma$ and $m,m'\in M_A$, if $Q::m::m'\in \sigma, Q::m'::m\not\in\sigma$ then $m\prec_\sigma m'$;
\item for any $m,m'\in M_A$ such that $m\not\prec_A m'$ and the last common enabler of $m$ and $m'$ is not an $\mv O$-move then $m\not\prec_\sigma m'$.
\end{enumerate}
\end{definition}
A strategy $\sigma$ is \emph{deadlock-free} if it has a precedence relation $\prec_\sigma$. From now on we only consider deadlock-free strategies. Deadlock-free strategies are interesting in their own right, and represent an alternative to conventional notions of termination (such as may, must or may-must) when timing is known. On the one hand termination analysis is simplified since the type (the arena) contains all the timing information. But on the other hand composition becomes more delicate as there is no room for the two strategies to wait for each other to perform certain common actions. Their synchronization needs to be on the nose, and for it to work at all it is essential that we rule out ``causal loops'' which may take an arbitrary amount of time to sort themselves out.  

For this reason we use precedence also on strategies. It is an order consistent with arena precedence ($\prec_A$) and preserved by alternative moves ($\asymp$), which prevents deadlock from two players waiting for each other.  The interesting requirement is the last one, which gives a term control over when it evaluates its arguments, which always have $\mv O$-moves as their last common enabler, but \emph{not} over when arguments it applies to functions are evaluated, which always have $\mv P$-moves as last common enablers. 

Let $P\keep A$ be a play from which all moves not in $A$ have been removed. Let $\sigma\keep A=\{P\keep A, P\in \sigma\}$. Let $\Sigma_{A,B,C}$ be the set of sequences over $M_A, M_B, M_C$. The interaction and composition of  strategies $\sigma:A\multimap B, \tau:B \multimap C$ are:
\begin{align*}
\sigma\,||\,\tau&=\{ P\in\Sigma_{A,B,C} \mid  P \keep A\multimap B\in\sigma\text{ and } P\keep B\multimap C\in \tau \}.
\\
\sigma;\tau &= \{ P\keep A\multimap C\mid P\in \sigma\,||\,\tau\}.
\end{align*}
Let the interleaving of two strategies $\sigma:A\multimap B, \tau:C\multimap D$ be the set
\[
\sigma\otimes\tau=\{ P\in\mathcal L(A\otimes C\multimap B\otimes D) \mid \\ P \keep A\multimap B\in\sigma\text{ and } P\keep C\multimap D\in \tau \}.
\]


\newcommand{\copycat}{c\!c}
\begin{definition}[Copycat]\label{def:cc}
  We define the \emph{copycat} 
  ${\copycat}_A:A\multimap A$, as the set of all
  plays $P$ such that for all $m\in M_A$,
  $\mathit{inl}(m)\in P$ if and only if $\mathit{inr}(m)\in P$. Moreover, 
  $\mathit{inr}(m)$ occurs before $\mathit{inl}(m)$ in $P$ if and only if $m$ is an $\mv O$-move in~$A$.
\end{definition}

\begin{theorem}\label{thm:games}
There exists a symmetric monoidal closed category with arenas $A$ as objects and strategies $\sigma: A\multimap B$ as morphisms where
\begin{enumerate}
\item identity $id_A:A \multimap A$ is the copycat strategy on $A$;
\item the tensor product is the disjoint union of arenas and interleaving of strategies, respectively;
\item the unit object is the empty arena (no moves);
\item the natural isomorphisms (associator, unitors, commutator) are (the obvious) copycat strategies;
\item currying is relabeling of moves in arenas induced by the obvious isomorphism between $A\otimes B\multimap C$ and $A\multimap B\multimap C$;
\item the morphism $eval_{A,B}:(A\multimap B)\otimes A\rightarrow B$ consists of two copycat behaviours between the $A$ and $B$ components, respectively.
\end{enumerate}
\end{theorem}
\subsection{Interpretation of PIA}
Let $\mathcal R$ be the discrete category with objects elements in $\mathbb N[\aff]$, $\catplus$ the additive operator of the semigroup semiring (Eq.~\ref{eq:add}), $\cattimes$ the multiplicative operator of the semigroup semiring (Eq.~\ref{eq:mul}), and the unit $0$ the additive unit (Eq.~\ref{eq:unit}). The associativity, distributivity and zero laws  follow from the semiring properties. 

Let $\mathcal G$ be the category of games from Thm.~\ref{thm:games}. The functor $\cdot$ is given in Def.~\ref{def:schedact} and it satisfies the associativity law given in Eqn.~\ref{eq:sro}. Any arena $A$ and schedules $J,K$ induce obviously isomorphic arenas
 $J\cdot A\,\otimes\, K\cdot A\simeq (J\,\catplus\, K)\cdot A$. The strategy $\delta_{J,K,A}$ is the one induced by the arena isomorphism and it satisfies all required coherence conditions. 

The arena of expressions (base type) is given by $\sbr{\lexpt}=(M,\tau,E,\lambda,\vdash,\asymp)$, where
\begin{align*}
M& =  \{\overline q, \overline a, q\}\cup \mathbb N, \\
\tau& =  \{\overline q\mapsto 0, \overline a\mapsto 1, q\mapsto 0\} \cup \mathbb N\times \{1\},\\
E &= \{\overline q, q\},\\
\lambda& =  \{\overline q\mapsto \mv{OQN}, \overline a\mapsto\mv{PAN}, q\mapsto\mv{OQM}\}\cup \mathbb N\times\{\mv{PAM}\},\\
{\vdash} &=  \{\overline q\vdash\overline a\}\cup \{q\}\times\mathbb N,\\
{\asymp} &= \{\overline q\asymp q\}\cup \{\overline a\}\times\mathbb N.
\end{align*}
In the concrete arena for expressions the initial question $q$ (or its alternative dummy $\overline q$) happen at 0 and the answer $i$ (or the alternative dummy $\overline a$) happen at 1. Note that we demand that an actual question receives an actual answer, not a dummy. 

The scheduling of commands $\lcompc_{x,y}:[x]\cdot\lcomt \multimap[y]\cdot\lcomt\multimap\lcomt$ is interpreted by the strategy consisting of the unique play in arena $\sbr{[x]\cdot\lcomt\multimap[y]\cdot\lcomt\multimap\lcomt}$ in which P does not play dummy moves unless responding to dummy O moves.

Operators $\lopc_{x,y}:[x]\cdot\lexpt \multimap[y]\cdot\lexpt\multimap\lexpt$ are interpreted by a strategy which is a set of plays, all with the same schedule as determined by their arena $\sbr{[x]\cdot\lexpt \multimap[y]\cdot\lexpt\multimap\lexpt}$, in which the final P-answer is calculated as the corresponding arithmetical operation applied to the preceding O-answers. 

Branching $\lifc_{\theta,x,y}:[x]\cdot\lexpt\multimap[y]\cdot\sigma\multimap[y]\cdot\sigma\multimap \sigma$, with $x\leq y$, is a strategy in arena
\[
\sbr{[x]\cdot\lexpt\multimap[y]\cdot\sigma\multimap[y]\cdot\sigma\multimap \sigma}
\]
defined as follows. The schedule constraint $x\leq y$ ensures that O answers first in the $\lexpt$ component, the guard. If it is not zero then a question is asked in the first $\sigma$ component and a dummy question in the second; the proper O-answer is then replicated as the final P-answer, while the dummy O-answer is ignored. Alternatively, if it is 0 then the question is asked in the second component and the dummy question in the first component, with the answer copied as final P-answer and the dummy answer ignored. 

The local variable binder $\lnewc_{\sigma,J,K}$ is interpreted in arena $\sbr{(J\cdot \lexpt\multimap K\cdot\lacct\multimap\sigma)\multimap\sigma}$ in the same way as the local-variable strategy is interpreted in IA, in a history-sensitive way: whenever P answers in the $\lexpt$ arena it is either with the same answer as the last answer in the $\lexpt$ arena or with the last O-answer in the $\lacct$ arena, whichever is most recent.  

\subsection{Technical proofs}\label{sec:tec}

In this section, we show some of the main intermediate results and proofs demonstrating that the concrete category of games $\mathcal G$ is well defined and satisfies the required properties. 

This ancillary lemma is useful for proving further results about strategies:
\begin{lemma}\label{lem:stratprefix}
  For any (deadlock-free) strategies $\sigma$ on $A\multimap B$ and $\tau$ on $B\multimap C$, and any positions $Q$ of $\sigma$, $R$ of $\tau$ such that $Q\keep B=R\keep B$,  then there is a sequence $P\in\sigma||\tau$ such that $Q$ is a prefix of $P\keep A\multimap B$ and $R$ is a prefix of $P\keep B\multimap C$.
\end{lemma}
\begin{proof}
  Let $\prec'$ be the minimum transitive relation on moves of $A$, $B$, and $C$ such that $m\prec'm'$ if $m\prec_\sigma m'$ or $m\prec_\tau m'$. Assume for contradiction that there are two moves $m$, $m'$ such that $m\prec_\sigma m'$ and $m'\prec_\tau m$; both moves would obviously have to be moves of $B$. Because $m'\prec_\tau m$, $m\nprec_B m'$. Because $m\prec_\sigma m'$, $m'\nprec_B m$. If $m$ and $m'$ do not have an $\mathsf{O}$-move as their last common enabler in $A\multimap B$, then because $m\nprec_{A\multimap B}m'$, we have $m\nprec_\sigma m'$, a contradiction. If they do have an $\mathsf{O}$-move as their last common enabler in $A\multimap B$, they must have the same move as their last common enabler in $B\multimap C$, where it is a $\mathsf{P}$-move, and thus because $m'\nprec_{B\multimap C}m$, $m'\nprec_\tau m$, also a contradiction. Thus, the assumption is wrong; and so, ${\prec'}$ is well-founded.

  Assume for contradiction that for given arenas $A$, $B$, $C$, that $Q$, $R$ are the counterexamples to the lemma with the largest total length. (They must be finitely long because the arenas contain finitely many questions, and all answers in a play must have a question enabling them.) Define $M$ as the set of all moves legal in the respective arenas in $Q$ or in $R$, or (transitively) justified by such a move. Choose a ${\prec'}$-least move $m\in M$ (such a move must be legal in $Q$ or $R$, because otherwise, its enabler would be ${\prec'}$-less than it, and there must be such a move or else $Q$ and $R$ are plays whose common moves appear in the same order and thus finding a suitable $P$ is trivial). Without loss of generality, assume that either $m$ is a move of $C$, or an $\mathsf{O}$-move of $B$ (the proof in the other cases is the same with $\sigma$ and $\tau$, and $Q$ and $R$, exchanged). Let $P_\tau$ be a play of $\tau$ with $R$ as a prefix in which $m$ or an alternative to it appears as early as possible. (Without loss of generality, assume that it is $m$ that appears.) If $R::m$ is not a prefix of $P_\tau$, then there must be a move $m'\in M$ immediately before $m$ in $P_\tau$; then $m'\nprec'm$ (by the definition of m), so $m'\nprec_\tau m$ (by the definition of ${\prec'}$), so $P_\tau$ with $m$ and $m'$ exchanged is a play of $\tau$ (by the definition of ${\prec_\tau}$), contradicting the assumption that $P_\tau$ is chosen such that $m$ appears as early as possible. Thus, $R::m$ is a prefix of $P_\tau$. If $m$ is a move of $C$, then we have $Q$ and $R::m$ as a counterexample to the lemma, violating the assumption that $Q$ and $R$ formed the counterexample with the largest total length. If $m$ is a move of $B$, then similarly we have $Q::m$ and $R::m$ as a counterexample to the lemma ($Q::m$ is a prefix of a play in $\sigma$ because $m$ is an $\mathsf{O}$-move of $A\multimap B$ and $\sigma$ is responsive), again violating the same assumption.

  Therefore, there cannot be a longest counterexample, and thus there cannot be any counterexample, to the lemma.
\end{proof}

To prove $\mathcal G$ a category, we need to show that it is closed under composition, that composition is associative, and that it has identities.

\begin{theorem}[Closure under composition]
  The composition $\sigma;\tau$ of two (deadlock-free) strategies $\sigma$ on $A\multimap B$, $\tau$ on $B\multimap C$ is a (deadlock-free) strategy on $A\multimap C$.
\end{theorem}
\begin{proof}
  We show that $\sigma;\tau$ is a set of plays on $A\multimap C$, and that it is responsive, saturated, and deadlock-free.

  $\sigma;\tau$ is by definition a set of sequences of moves of $A\multimap C$. For each sequence:
  \begin{itemize}
    \item All moves must be distinct, because two identical moves from $A$ would imply there were two identical moves in $\sigma$, and likewise for $C$ and $\tau$. Similar arguments shows that there is one move from each set of alternatives, and that each question enables exactly one answer.
    \item All moves must be either initial, or enabled by an earlier move in the sequence: 
      \begin{itemize}
      \item By construction of the arenas, moves of $A$ cannot enable moves of $B$ or C in the original arenas $A\multimap B$, $B\multimap C$, nor can moves of $C$ enable moves of $A$, nor moves of $B$ enable moves of $C$.
      \item Moves of $A$ enabled by other moves of $A$, and moves of $C$ enabled by other moves of $C$, in the initial arenas, will have both the move and enabler included in the same order in $\sigma;\tau$.
      \item Initial moves of $C$ are initial moves of both $B\multimap C$ and $A\multimap C$ and so cannot be enabled in either arena.
      \item Initial moves of $A$ (the only remaining case) are enabled by each initial move of $C$, and so are enabled in $A\multimap C$ by whichever initial move of $C$ happens to be included in the relevant play of $\tau$.
      \end{itemize}
    \item The order of moves in $\sigma;\tau$ must be consistent with ${\prec}_{A\multimap C}$; the timings must be in non-decreasing order, because otherwise either the order would be inconsistent with $\prec_{A\multimap B}$ or $\prec_{B\multimap C}$ respectively in a play of $\sigma$ or $\tau$, or else $B$ is the null arena (and thus no play of $\sigma$ contains any moves); and no moves can answer a move that is enabled by a move answered earlier in a sequence, using a similar argument to the above.
  \end{itemize}
  Therefore, $\sigma;\tau$ is a set of plays on $A\multimap C$.

  To see that $\sigma;\tau$ is responsive, consider a position $Q$ of $\sigma;\tau$, and an $\mathsf{O}$-move $m$ legal in $Q$. Without loss of generality, assume that $m$ is a move of $A$ (the proof for $m$ a move of $C$ is similar). Let $P$ be an element of $\sigma||\tau$ such that $Q$ is a prefix of $P\keep A\multimap C$, $P'$ the shortest prefix of $P$ where $Q = P'\keep A\multimap C$, $Q_\sigma$ be $P'\keep A\multimap B$, $Q_\tau$ be $P'\keep B\multimap C$. Then because $\sigma$ is responsive, $Q_\sigma::m$ is a position of $\sigma$; $Q_\tau$ is a position of $\tau$ by definition; and so by Lem.~\ref{lem:stratprefix}, $Q::m$ is the prefix of some play in $\sigma;\tau$.

  For deadlock-freedom, we need to prove the existence of a $\prec_{\sigma;\tau}$. We claim that $\prec'$ defined in the proof of Lem.~\ref{lem:stratprefix} meets all the requirements to be such a relation. The first three requirements of Def.~\ref{def:precs} are obvious, and the last requirement is trivially met because the set of last common enablers of an $A$-move and $C$-move are the initial questions of $C$ (which contains only $\mathsf{O}$-moves, and is nonempty except in the degenerate case where $\sigma$ has no nonempty plays), so we need only prove that if $m\nprec'm'$ and $Q::m::m'$ is a position of $\sigma;\tau$, then $Q::m'::m$ is also a position of $\sigma;\tau$. Let $P$ be an element of $\sigma||\tau$ such that $Q::m::m'$ is a prefix of $P\keep A\multimap C$; let $P'$ be the shortest prefix of $P$ such that $P'\keep A\multimap C = Q::m::m'$; and let $P''$ be the longest prefix of $P$ such that $P''\keep A\multimap C = Q$. Then let $Q_\sigma = P''\keep A\multimap B$, $Q_\sigma::R_\sigma = P'\keep A\multimap B$, and likewise for $Q_\tau$ and $R_\tau$. Let $N = \{n\in R_\sigma\cup R_\tau | m\prec'n\}\cup\{m\}$, and $N' = \{n'\in R_\sigma\cup R_\tau | n'\not\in N\}$. $Q_\sigma::R_\sigma$ is a position of $\sigma$, and because $\sigma$ is deadlock-free and $n\nprec' n'$ implies $n\nprec_\sigma n'$, it must be possible to repeatedly exchange the positions of a move of $N$ in $R_\sigma$ and an immediately following move of $N'$ in $R_\sigma$ and still have a position of $\sigma$; likewise for $\tau$. The rearrangement is the same in both strategies, and so the rearranged $Q_\sigma::R_\sigma$ and $Q_\tau::R_\tau$ have their common moves in the same order. $m\in N$ by definition, $m'\in N'$ because $m\nprec'm'$ by assumption. And therefore, via Lem.~\ref{lem:stratprefix}, $Q::m'::m$ is a prefix of a play of $\sigma;\tau$.

To prove that $\sigma;\tau$ is saturated, we need to prove that if $m$ is an $\mathsf{P}$-move and/or $m'$ is a
 $\mathsf{O}$-move, and $P::m::m'::P'\in \sigma$, then if $P::m'::m::P'$ is a play $P::m'::m::P'\in \sigma$. The proof is
 along similar lines to the previous proof. Define ${\prec_{\mathit sat}}$ such that $o {\prec_{\mathit sat}} p$ for every
 $\mathsf{O}$-move $o$ and $\mathsf{P}$-move $p$ in $A$, $B$, and $C$. Let $P_{\sigma||\tau}$ be the element of
 $\sigma||\tau$ such that $P_{\sigma||\tau}\keep A\multimap C = P::m::m'::P'$; let $P'_{\sigma||\tau}$ be the longest
 prefix and $P''_{\sigma||\tau}$ the longest suffix of $P_{\sigma||\tau}$, such that $P'_{\sigma||\tau}\keep A\multimap C = P$
 and $P''_{\sigma||\tau}\keep A\multimap C = P'$, and define $R$ such that $P'_{\sigma||\tau}::R::P''_{\sigma||\tau} = P_{\sigma||\tau}$. 
Then let $P'_\sigma = P'_{\sigma||\tau}\keep A\multimap B$, $P''_\sigma = P''_{\sigma||\tau}\keep A\multimap B$,
 $R_\sigma = R\keep A\multimap B$, and likewise for $P'_\tau$, $P''_\tau$, and $R_\tau$. Let $N = \{n\in R_\sigma\cup R_\tau | m\prec_{\mathit sat}n\}\cup\{m\}$, and $N' = \{n'\in R_\sigma\cup R_\tau | n'\not\in N\}$.
 $P'_\sigma::R_\sigma::P''_\sigma$ is a position of $\sigma$, and because $\sigma$ is saturated and $n\nprec_{\sigma||\tau}n'$ implies $n$ is a $\mathsf{P}$-move and/or $n'$ is an $\mathsf{O}$-move, it must be possible to repeatedly
 exchange the positions of a move of $N$ in $R_\sigma$ and an immediately following move of $N'$ in $R_\sigma$ and still
 have a play of $\sigma$; likewise for $\tau$. As such, applying the same rearrangement to $P_{\sigma||\tau}$ leads to 
an interaction which forms a play of $\sigma$ if restricted to moves of $A\multimap B$, and a play of $\tau$ if restricted
 to moves of $B\multimap C$. And thus, applying the same rearrangement to $P::m::m'::P'$ gives a play of $\sigma;\tau$,
 $P::m'::m::P'$.

  Therefore, $\sigma;\tau$ is a strategy on $A\multimap C$.
\end{proof}

\begin{theorem}[Associativity]
  $(\sigma;\tau);\upsilon$ = $\sigma;(\tau;\upsilon)$.
\end{theorem}
\begin{proof}
  Define the three-way interaction $\sigma||\tau||\upsilon$ of
strategies $\sigma:A\multimap B$, $\tau:B\multimap C$, $\upsilon:C
\multimap D$, as $\{P\in\Sigma_{A,B,C,D}|P\keep A\multimap B=\sigma
\wedge P\keep B\multimap C=\sigma\wedge P\keep C\multimap D=\sigma\}$.
$(\sigma;\tau)||\upsilon = (\sigma||\tau||\upsilon)\keep A\cup C\cup D$,
because for each element of $\sigma;\tau$, there is by definition an
element of $\sigma||\tau$ corresponding to it. For the same reason,
$\sigma;(\tau||\upsilon) = (\sigma||\tau||\upsilon)\keep A\cup B\cup D$.
Thus, $(\sigma;\tau);\upsilon = (\sigma||\tau||\upsilon)\keep A\multimap
D = \sigma;(\tau;\upsilon)$.
\end{proof}

\begin{theorem}[Identity]
  The copycat strategy $id_A$ for any arena $A$ is in fact a strategy, and a left and right identity under strategy composition.
\end{theorem}
\begin{proof}
  By definition, $id_A$ is a set of plays on $A\multimap A$.

  Because $id_A$ is defined as containing all plays except those where specific $\mathsf{O}$-moves appear after specific $\mathsf{P}$-moves, $id_A$ is trivially both saturated and responsive (the requirement is preserved by moving $\mathsf{O}$-moves earlier or $\mathsf{P}$-moves later, and cannot prevent $\mathsf{O}$-moves appearing unless they have already appeared in the play).

  To show deadlock-freedom, a suitable $\prec_{id_A}$ is the least transitive relation where $in_x(m) \prec_{id_A} in_y(m')$ for $x,y\in{l,r}, m \prec_A m'$, and where $in_x(m) \prec in_y(m)$ with $x\neq y\in{l,r}$ and $in_x(m)$ an $\mathsf{O}$-move. This relation is obviously well-founded, obviously respects precedence on the arena and alternatives, and obviously lists all pairs of moves that cannot be reversed. It also obeys the last common enabler rule, because the last common enablers of $in_l(m)$ and $in_r(m')$ are the initial moves of $id_A$, which are $\mathsf{O}$-moves (except in the degenerate case where $A$ has no initial moves, whose identity contains no moves in its plays and thus is trivially deadlock-free).

  To prove $id_A;\sigma=\sigma$, consider the interaction $id_A||\sigma$. (This contains moves from two distinct copies of $A$; we label them $A_1$ and $A_2$ for clarity, with $id_A$ on $A_1\multimap A_2$ and $\sigma$ on $A_2\multimap B$.) By the definition of $id$, in the interaction $id_A||\sigma$, for each move of $A_2$ there is a move of $A_1$ and vice versa; and the $\mathsf{O}$-moves come first. Thus, for each play of $id_A;\sigma$, there is a play of $\sigma$ that contains the same moves (but not necessarily in the same order). However, the only changes to the ordering of the moves that are made are to move $\mathsf{P}$-moves later and/or $\mathsf{O}$-moves earlier. Thus, $id_A;\sigma\subseteq\sigma$. Additionally, by replacing each $\mathsf{O}$-move $o$ with $in_{A_1}(o)::in_{A_2}(o)$ and each $\mathsf{P}$-move $p$ with $in_{A_2}(p)::in_{A_2}(p)$ in any play of $\sigma$, the resulting sequence is clearly an element of $id_A||\sigma$, and the play derived from it is clearly identical to the original play. Thus, $\sigma\subseteq id_A;\sigma$. And so, $id_A;\sigma=\sigma$. A similar argument can be used to prove that $\sigma;id_A=\sigma$.
\end{proof}

\begin{theorem}
  Taking ${\otimes}$ on strategies to be interleaving of strategies, $(\sigma;\sigma')\otimes(\tau;\tau') = (\sigma\otimes\tau);(\sigma'\otimes\tau')$.
\end{theorem}
\begin{proof}
  We can decompose this condition into three simpler conditions, $(\sigma\otimes id_B);(\sigma'\otimes id_B)=(\sigma;\sigma')\otimes id_B$, $(id_A\otimes\tau);(id_A\otimes\tau')= id_A\otimes(\tau;\tau')$, and $(\sigma\otimes id_B);(id_A\otimes\tau)=(\sigma\otimes\tau) = (id_A\otimes\tau);(\sigma\otimes id_B)$. Each of these conditions becomes obvious upon replacing $;$ and $\otimes$ with their definitions.
\end{proof}
\begin{theorem}[Unit object]
  With $I$ as the empty arena, $A\otimes I\simeq A\simeq I\otimes A$ for all arenas $A$.
\end{theorem}
\begin{proof}
  $I$ has no moves, so its disjoint union with any arena is isomorphic to that arena.
\end{proof}
This proves that $\mathcal G$ is a monoidal category. Proving it to be also symmetric and closed requires proving several coherence constraints, but each of these constraints are requirements that relabelings are natural isomorphisms (which is obviously true), or that relabelings of the identity commute (which is also obviously true).

\allowdisplaybreaks 
The fact that $\cdot$ is a proper functor is immediate. We  prove that Eqn.~\ref{eq:sro} holds for $\mathbb
N[\mathrm{Aff}_1^c]$ and $\mathcal G$:
\begin{theorem}
  $(J\cattimes K)\cdot\sigma = J\cdot(K\cdot\sigma)$ for $J,K \in \mathbb
N[\mathrm{Aff}_1^c]$ and $\sigma$ a strategy.
\end{theorem}
\begin{proof}
  \begin{align*}
    J\cdot(K\cdot\sigma) &= J\cdot\biguplus_{x\in
\mathrm{Aff}_1^c}\biguplus_{n\leq K(x)} x\cdot A\\
    &= \biguplus_{x,y\in\mathrm{Aff}_1^c}\biguplus_{m\leq
J(x)}\biguplus_{n\leq K(y)} x\cdot (y\cdot A)\\
    &= \biguplus_{x,y\in\mathrm{Aff}_1^c}\biguplus_{n\leq J(x)K(y)} (x
\times y)\cdot A\\
    &= \biguplus_{x\in\mathrm{Aff}_1^c}\biguplus_{n\leq (J\cattimes K)(x)}
x\cdot A
  \end{align*}
\end{proof}

\section{Conclusion}

We have presented a bounded affine type system using an abstract resource semiring and gave a generic type inference and coherent categorical semantics for it. To illustrate its flexibility we used it to give a precise timing discipline to a recursion-free functional programming language with local state defined using a game-semantic model. The first, more theoretical, part of the paper is motivated by our desire to generalize our previous work on resource-sensitive type systems (such as SCC) and the results should be broadly applicable to many such systems. The second part is a highly non-trivial motivating application of the theory where schedules of execution are treated as a resource, and is driven by our interest in enhancing the \emph{Geometry of Synthesis} hardware compiler with transparent, automatic pipelining. It is hopefully obvious that the use of a generic type system and categorical semantics imposes a high degree of abstract discipline which is essential in managing a complex type system and its interpretation in a correspondingly complex semantics.

For future work, carefully injecting some data-dependencies into the type system would be highly desirable. Full-blown data dependency, especially in the presence of recursion, would make automatic type inference unfeasible. This goes beyond a mere decidability result. In type inference we outsource the heavy lifting to an external SMT solver and, so long as it can \emph{attempt} to solve the associated system of constraints with a decent chance of success (as determined by practical experiments) we are content. But when failure of inference due to computability issues is a matter of course (see e.g.~\cite{DBLP:conf/popl/LagoP13}) then it means that the type system is overly ambitious. Fortunately there is room for an interesting middle ground. To stay in the concrete context of precise timing, access to resources can be data dependent in (logically) simple ways even in the absence of recursion. An example is that of caching behavior: requesting an item of data the second consecutive time can be accomplished much faster than the first time around. 

The game semantics of Sec.~\ref{sec:gamhort} introduced a number of innovations which deserve to be studied in more depth.We did not attempt to prove (or even formulate) \emph{definability} in timed games, which is an interesting question. Also, although our game model is formulated for the concrete programming language directly, it is quite clear that much of its formulation is independent of the particular choice of resource semiring. The only place where the choice of the resource semiring (schedules) is important is in the Arena definition, Def.~\ref{def:arena}(\ref{def:arenat}), in which move ordering needs to be consistent with timing. A relaxation of this rule may lead to a generic  game model of the abstract type system. 

Finally, an efficient implementation of the pipelining mechanism in the hardware compiler requires the exploration of several possible ways in which detailed knowledge of timing can be exploited. The current implementation of the hardware compiler\footnote{See \url{http://veritygos.org}} is not compatible with pipelining because the circuit implementing contraction ($\delta$) can only be used sequentially. The new scheduled contraction operator $\delta_{J,K}$ on the other hand can be used concurrently and can be given a finite-state implementation. The sizes of schedules ($J,K$) is known and finite and so is the order in which signals arrive, therefore their order can be used to determine signal routing. 

On the other hand, the timing information at our disposal is now much richer than simply knowing the order of events in the pipelines. We have full knowledge of the timing of each event; our timing is relative, but computing absolute timings from the relative timing information is quite easy. This means that our locally-synchronous-globally-asynchronous handshake protocol between components can be replaced by a globally-synchronous communication paradigm. Control signals indicating when data is available are now redundant, since this information is available at compile-time. Removing the handshake infrastructure is an interesting and appealing idea, but it is difficult to predict if it will lead to any performance improvements, since a new global clocking infrastructure needs to replace it. We will examine these questions in the near future.

\paragraph{Acknowledgment.}
Sec.~\ref{sec:cf} benefited significantly from discussions with Steve Vickers. Olle Fredriksson and Fredrik Nordvall-Forsberg provided useful comments. The authors express gratitude for their contribution. 

\bibliographystyle{apalike}
\bibliography{manual}

\end{document}